%% file: FusedLasso.tex
\documentclass[12pt,english]{article}
\usepackage[T1]{fontenc}
\usepackage[latin9]{inputenc}
\usepackage{geometry}
\usepackage{amsmath}
\usepackage{graphicx}
\usepackage[authoryear]{natbib}
\usepackage{amsfonts}
\bibpunct[; ]{(}{)}{,}{a}{}{;} 
\makeatletter

\usepackage{amsthm}
\newtheorem{thm}{Theorem}
\newtheorem{lem}{Lemma}
\newtheorem{prop}{Proposition}
\newtheorem{definition}{Definition}
\usepackage[ruled]{algorithm2e}
\usepackage{multirow}
\usepackage{dsfont}
\usepackage{babel}
\usepackage{enumerate}

\makeatother

\include{definitions}

\begin{document}

\title{A coordinate-wise optimization algorithm for the Fused Lasso}

\author{Holger H\"ofling \and Harald Binder \and Martin Schumacher}

\maketitle
\begin{abstract}
$L_1$-penalized regression methods such as the Lasso \citep{Tibshirani1996} 
that achieve both variable selection and shrinkage have been very popular. An extension
of this method is the Fused Lasso \citep{TW2007}, which allows for the incorporation
of external information into the model. In this article, we develop new and fast
algorithms for solving the Fused Lasso which are based on coordinate-wise
optimization. This class of algorithms has recently been 
applied very successfully to solve $L_1$-penalized problems very quickly
\citep{FHT2007}. As a straightforward coordinate-wise procedure does not
converge to the global optimum in general, we adapt it in two ways, using 
maximum-flow algorithms and a Huber penalty based approximation to the loss
function. In a simulation study, we evaluate
the speed of these algorithms and compare them to other standard methods. As
the Huber-penalty based method is only approximate, we also evaluate
its accuracy. Apart from this, we also extend the Fused Lasso to logistic as well as proportional hazards
models and allow for a more flexible penalty structure.
\end{abstract}

\section{Introduction}
$L_1$ penalization methods have been used very successfully for variable selection as well
as prediction. One of the most widely known methods is the Lasso \citep{Tibshirani1996} which
applies an $L_1$ penalty to the coefficients in the model and minimizes the loss function
\[
\frac{1}{2}(\by -\bX\bbeta)^T(\by - \bX\bbeta) + \lambda_1 \sum_{k=1}^p w_k|\beta_k|
\]
where $\bX \in \mathds{R}^{n \times p}$ as well as $\by, \bbeta \in \mathds{R}^p$. Here,
$w_k > 0$ is some weight associated with coefficient $\beta_k$. By using this
penalty, it is possible to fit models even if $p > n$, perform variable selection and 
improve prediction performance through shrinkage. In recent years, this model has been
extended in many ways, for example by allowing for grouping of variables 
\citep[e.g.][]{Yuan2006}, by including an additional ridge regression penalty 
\citep[see][]{Zou2005} or by choosing appropriate weights for the $L_1$ penalty to
improve its asymptotic properties \citep{Zou2006}, to name just a few.

Often, there is also some other prior information
that could be incorporated into the model. An extension of the Lasso that takes advantage of such
information is 
the Fused Lasso \citep{TW2007}. In its simplest form, it assumes that the coefficients
$\beta_k$ are ordered and imposes an additional $L_1$ penalty on the differences of 
neighbors. The loss function that is being minimized in this case is then
\[
\frac{1}{2}(\by -\bX\bbeta)^T(\by - \bX\bbeta) + \lambda_1 \sum_{k=1}^p w_k |\beta_k| + \lambda_2 
    \sum_{k=1}^{p-1} w_{k,k+1}|\beta_{k+1} - \beta_k|,
\]
where $\lambda_2$ is the parameter that controls the smoothness of the resulting solution.
Most often, this model has been used in the form of the Fused Lasso Signal Approximator (FLSA)
where we have that $\bX = \bI$, the identity matrix. This model can, for example, be
used for Comparative Genomic Hybridization (CGH) data \citep[see][]{TW2007}. For CGH data, at various points of a
chromosome, a noisy estimate of the number of available copies of the genomic data is 
being taken. As duplications or deletions usually occur over large areas, the $\lambda_2$
penalty smoothes the copy number estimate over neighboring data points, thereby 
increasing the accuracy of the estimation. Another 
possible application of the FLSA is the denoising of images \citep[see][]{FHT2007}.
Recently, some theoretical properties of the model have been shown by \citet{Rinaldo2009}
and a path algorithm has been developed in \citet{Hoefling2010}. 

In the above model, it is not necessary to restrict the $\lambda_2$-penalty
to neighboring coefficients. To allow for greater flexibility, we penalize
all differences of coefficients that correspond to an edge in a graph. To be more
precise, let $V= \{1, \ldots, p\}$ be a set of nodes, each of which corresponds to 
one of the coefficients. Then we define a graph $\mathcal{G} = (V, E)$ with 
edge set $E$ and penalize all differences $|\beta_k - \beta_l|$ for which
$(k,l) \in E$, i.e. for which the graph $\mathcal{G}$ has a corresponding edge.
Additionally, we also allow for different weights for each of the edges (and each
of the $|\beta_k|$ penalties) so that we get the loss function
\begin{equation}
g(\bX, \by, \bbeta, \bw, \mathcal{G}, \blambda) = 
    \frac{1}{2} (\by -\bX\bbeta)^T(\by - \bX\bbeta) + 
    \lambda_1 \sum_{k=1}^p w_k|\beta_k| + \lambda_2 
    \sum_{(k,l) \in E, k < l} w_{k,l}|\beta_{k} - \beta_l|.
    \label{Eq:FusedLasso}
\end{equation}
In addition to the squared error loss function above, we can of course also 
include other convex functions. Specifically, we will also look at
logistic regression as well as the Cox proportional hazards models. 

To the best of our knowledge, there is currently no optimized algorithm available
that minimizes the loss function in Equation (\ref{Eq:FusedLasso}). 
However, solving this model efficiently is important so that it can be applied in
larger datasets or in situations where a large number of fits of a model is important 
(e.g. when bootstrapping a method that uses cross-validation). A method that
has been very successfully used for the Lasso model is coordinate-wise 
optimization \citep[][see]{FHT2007}. In this article, we will present two algorithms
that extend this technique to the Fused Lasso model, one exact, the other approximate. 
In Section \ref{Sec:Algorithm}, we will outline
the algorithms and prove that it converges to a global optimum for the exact method. Its speed 
will be evaluated and compared to other methods in Section \ref{Sec:Simulations}.
The results will then be discussed in Section \ref{Sec:Discussion}. An implementation
of this method is available from \texttt{CRAN} as \texttt{FusedLasso} \citet{Hoefling:FusedLassoPackage}.

\section{\label{Sec:Algorithm}The Algorithm}
In this section we want to present our algorithms for solving the Fused Lasso
 problem specified in Equation (\ref{Eq:FusedLasso}). We will do this in several steps. 
First, we will
present the naive implementation of the coordinate-wise method that 
just optimizes each coordinate of the function $g$ iteratively until it
converges. However, this solution is not guaranteed to be the global optimum or even  
be unique and in general performs rather poorly in our setting (see Section \ref{Sec:Naive}). 
In order to avoid these problems, we force certain variables to be equal to each other (i.e. \emph{fused}) and thus
allow the coordinate-wise algorithm to overcome its limitations (see Sections
\ref{Sec:FusedSets} and \ref{Sec:DefineFusedSets}). To be more precise, for a set of 
\emph{fused} variables $F \subset \{1, \ldots, p\}$, we require
that $\beta_k = \beta_l$ for all $k, l \in F$. There can also be several, non-overlapping
sets $F$. These restrictions are then included in the loss function $g$ to form a 
constrained loss function $\tilde{g}$, which is then 
being optimized using the naive coordinate-wise procedure. By choosing
the sets correctly, our algorithm will converge to the global optimum, for which we provide
the proof in the appendix.

However, for models for which the number of variables is very large, the method used here can
be relatively slow. Therefore, we also introduce an approximate method that addresses this problem
and will be presented in Section \ref{Sec:Huber}.

\subsection{\label{Sec:Naive}Naive Algorithm}
A naive coordinate-wise optimization procedure for loss function (\ref{Eq:FusedLasso}) works analogous
to the method for the Lasso described in \citet{FHT2007}. For every $k \in \{1, \ldots, p\}$, 
$\beta_k$ is being optimized while keeping all other coefficients constant. This is
very simple as $\partial g(\bX, \by, \bbeta, \bw, \mathcal{G}, \blambda) / \partial \beta_k$ is
piecewise linear, increasing step function with positive jumps. For these,
the unique root can be found very easily. This simple procedure can be sped up considerably
for cases where $p > n$ as then most coefficients will always be $0$ and in general will not be changed
by the coordinate optimization steps. In order to exploit this, an active set $\mathcal{A}$
is defined and the coordinate-wise procedure optimizes only over the variables in 
$\mathcal{A}$, leaving out many variables that are $0$ and do not need to be considered in 
every iteration. In order to guarantee convergence, it is then necessary to 
regularly adapt the set $\mathcal{A}$ so that
variables that could become non-zero in the next step are included and variables that
have become 0 are excluded. As this technique is not a central point of this article, 
we will not elaborate on it any further. The
conditions for changing the set $\mathcal{A}$ can be found in Algorithm 
\ref{Alg:NaiveCoordinate} and are easily derived in closed form. Furthermore, the underlying concept 
of active variables is explained in more detail in \citet{FHT2007}.

\begin{algorithm}
    \KwData{$\by$, $\bX$, $\bbeta_0$}
    \KwResult{$\bbeta$ optimizing $g(\bX, \by, \bbeta, \bw, \mathcal{G}, \blambda)$}

    $\bbeta := \bbeta_0$\;
    $\mathcal{A} := \{k | \beta_k \neq 0\}$ \;
    \Repeat{$\mathcal{A}$ did not change}{
        \Repeat{converged}{
            \ForEach{$k \in \mathcal{A}$}{
                minimize $g(\bX, \by, \bbeta, \bw, \mathcal{G}, \blambda)$ over $\beta_k$\;
            }
        }
        Set $\mathcal{A} := \mathcal{A} \cup \left\{k |\beta_k=0  \text{ and } \text{argmin}_{\beta_k} g(\bX, \by, \bbeta, \bw, \mathcal{G}, \blambda) \neq 0
            \right\}$\;
    }
    \caption{Naive coordinate-wise optimization algorithm for the Fused Lasso.}
    \label{Alg:NaiveCoordinate}
\end{algorithm}

Our first concern here is of course that the coordinate-wise algorithm converges at all. The following
theorem shows that while a unique point does not necessarily exist, any cluster point of the 
algorithm is a coordinate-wise minimum of $g$.
\begin{thm}
Assume that the coordinate-wise algorithm starts at $\bbeta^{0}$. 
Let $\bbeta^{r}$ be the sequence generated by the coordinate steps for $r \in \mathds{N}$. 
Then $\bbeta^{r}$ is bounded and every cluster point 
is a coordinate-wise minimum of $g$.
\label{Thm:NaiveCoordinate}
\end{thm}
The proof can be found in the appendix. For the rest of the article, 
we will view a whole run of the coordinate-wise algorithm as
a single step in our procedure. Accordingly, $\bbeta^{(m)}$ will refer to the
estimate of $\bbeta$ after the $m$-th application of the coordinate-wise
minimization. 

The problem with this naive implementation of the coordinate-wise algorithm is that
it is not guaranteed to converge to a global optimum. Very often, it will be stuck
in local coordinate-wise minima although an improvement of the loss function would be possible if
several $\beta_k$ could be moved at the same time instead of separately. In order
to incorporate this, we will \emph{fuse} coefficients. When we refer to 2 
neighboring coefficients as \emph{fused}, we mean that we
force them  to be equal in the coordinate optimization steps and thus move together.
More specifically, we define a set of coefficients to be a fused set as follows:

\begin{definition}\label{Def:FusedSets}
Let $\mathcal{G} | F$ be the graph $\mathcal{G}=(V,E)$ restricted to the set $F$, i.e. 
$\mathcal{G} | F = (F, \{(k,l)|(k,l) \in E \text{ for } k,l \in F\})$. Then a set 
$F \subseteq \{1, \ldots, p\}$ is called a fused set w.r.t. $\bbeta$ if 
$\mathcal{G}|F$, is connected and $\beta_k = \beta_l \; \forall \; k,l \in F$.
\end{definition}

Now, we will incorporate these fusions into our loss function.

\subsection{\label{Sec:FusedSets}Incorporating the fused sets}
Let $\mfF$ be a collection of fused sets, i.e. $\mfF = \{ F_1, \ldots, F_{|\mfF|}\}$ where
each $F_i$ is a fused set according to the Definition \ref{Def:FusedSets} above. Furthermore, we assume
that the $F_i$ are non-overlapping and cover all coefficients, that is
\[
F_i \cap F_j \; \forall i \neq j \quad \text{and} 
\quad \cup_{i=1}^{|\mfF|} F_i = \{1, \ldots, p\}
\]
Then, given these sets we want to solve the constrained optimization problem
\begin{align*}
    &\text{minimize } g(\bX, \by, \bbeta, \bw, \mathcal{G}, \blambda) \\
    &\text{subject to } \beta_k = \beta_l \; \forall k,l \in F_i \; \forall F_i \in \mfF.
\end{align*}
However, as can be seen very easily, this is actually equivalent to 
\[
\text{minimize } g(\tilde{\bX}, \tilde{\by}, \tilde{\bbeta}, \tilde{\bw},
        \tilde{\mathcal{G}}, \blambda)
\]
where we define $\tilde{\bX}, \tilde{\by}, \tilde{\bbeta}, \tilde{\bw},
      \tilde{\mathcal{G}}$ as follows:
\begin{description}
    \item[$\tilde{\bX}$:] Here, column $i$ of $\tilde{\bX}$ corresponds to the sum
        of the columns of $\bX$ of the indices in $F_i$, i.e.
        \[
        \tilde{x}_{ij} := \sum_{k \in F_i} x_{kj} \quad \forall \; F_i \in \mfF.
        \]
    \item[$\tilde{\by}$:] Similar to above, we define 
    \[
        \tilde{y}_i := \sum_{k \in F_i} y_k \quad \forall F_i \in \mfF.
    \]
    \item[$\tilde{\bbeta}$:] $\tilde{\beta}_i = \beta_k$ for any $k \in F_i$ for all
    $F_i \in \mfF$.
    \item[$\tilde{\bw}$:] For the weights for $\lambda_1$ we set for each $F_i \in \mfF$
    \[
    \tilde{w}_i := \sum_{k \in F_i} w_k
    \]
    and for the weights associated with $\lambda_2$ and $\mathcal{G}$ define
    for each combination $F_i, F_j \in \mfF$ for $i \neq j$ that
    \[
    \tilde{w}_{i,j} := \sum_{(k,l)|(k,l) \in E; k \in F_i, l \in F_j} w_{k,l}.
    \]
    \item[$\tilde{\mathcal{G}}$:] Here we define the graph as $\tilde{\mathcal{G}} = 
        (\tilde{V}, \tilde{E})$ where 
    \begin{align*}
    \tilde{V} = \{1, \ldots, |\mfF|\} \quad \text{and} \\ 
    \tilde{E} = \{(i,j)|\exists k\in F_i, l\in F_j \text{ where } F_i, F_j \in \mfF,
        i \neq j \text{ and } (k,l) \in E\}.
    \end{align*}
\end{description}
It is easy to derive these adjusted input variables by replacing $\beta_k$ by $\tilde{\beta}_i$
for all $k \in F_i$ for every $F_i \in \mfF$ in Equation (\ref{Eq:FusedLasso}) 
and then simplify the result.
Therefore, we will omit a detailed derivation here.

Using these fused input data, we can now apply the coordinate-wise algorithm to $\tilde{g}$ until
it converges. However, choosing the correct fusions is crucial for the algorithm
to converge to the global optimum. In the definition of fused sets, we stated
that is is necessary that they are connected w.r.t. $\mathcal{G}$ and all
have the same $\beta$-value. However, this does not imply that all coefficients
with the same value that are connected have to be fused. This gives us some freedom
in choosing the connected sets. Next, we will specify how to choose them in order to
ensure that the algorithm converges to the global minimum.

\subsection{\label{Sec:DefineFusedSets} How to define the fused sets using a maximum-flow problem}
In the coordinate-wise algorithm, the issue that may prevent the algorithm
from converging to the correct solution is that, while it may not be possible to
improve the loss function by just moving one coefficient, it may be possible
to get an improvement by moving several at the same time. We assume that
we have some current estimate $\bbeta^{(m)}$ in step $m$, and we now want to identify
sets $F_i$ that can be moved in the next iteration of the coordinate-wise
algorithm. 

As a first step, we define the sets as all coefficients that 
are connected and have the same value. This is easily done by using the following 
equivalence relation $\sim$:
\begin{align*}
k \sim l \quad \Longleftrightarrow \quad &\beta_k^{(m)} = \beta_l^{(m)} \text{ and }
k \text{ is connected to } l \text{ w.r.t. } \mathcal{G} \\ 
    &\text{ only through nodes with value equal to } \beta_k^{(m)}.
\end{align*}
Here, we suppress for notational convenience the dependence of the relation
on $\mathcal{G}$ as well as $\bbeta^{(m)}$.
This equivalence relation is clearly reflexive, symmetric and transitive and therefore
partitions the set $\{1, \ldots, p\}$ into sets 
$P_1, \ldots, P_{|\sim|}$. We call the collection of these sets $\mfP$, i.e. 
\[
  \mfP = \{P_1, P_2, \ldots, P_{|\sim|} \}.
\]
For each of these sets $P_i$, all pairs of coefficients $k,l \in P_i$ are then
by definition of the equivalence relation connected in $\mathcal{G}$ and
have the same value of the coefficients (i.e. $\beta_k^{(m)}=\beta_l^{(m)}$). 
Furthermore, they are maximal in the sense that they
could not be enlarged and still have these properties.

This partition is the starting point for our fused sets. We now want to
specify in which situation these sets can stay as they are and when they have to be
split into several parts. The latter situation occurs when such a split makes an 
improvement of the loss function possible. For this, we need to look 
at the effect of the $\lambda_2$-penalty inside the groups $P_i$ separately.

In order to do this, we slightly rewrite the loss function. In the following, please note that any conditions
that are separated by commas are ``and'' conditions.
\begin{align*}
g(\bX,\by, &\bbeta, \bw, \mathcal{G}, \blambda) = 
    \frac{1}{2} (\by -\bX\bbeta)^T(\by - \bX\bbeta) + 
    \lambda_1 \sum_{k=1}^p w_k|\beta_k| + \lambda_2 
    \sum_{(k,l) \in E, k < l} w_{k,l}|\beta_{k} - \beta_l| = \\
    &\frac{1}{2} (\by -\bX\bbeta)^T(\by - \bX\bbeta) + 
    \lambda_2 \sum_{i=1}^{|\mfP|} \sum_{(k,l) \in E, k \in P_i, l \not \in P_i, k < l} 
        w_{k,l}|\beta_{k} - \beta_l| + \\
    &\lambda_1 \sum_{i=1}^{|\mfP|} \sum_{k \in P_i} w_k |\beta_k| +  
    \lambda_2 \sum_{i=1}^{|\mfP|} \sum_{(k,l) \in E, k,l \in P_i, k < l} 
        w_{k,l}|\beta_{k} - \beta_l| = \\
     &h(\bX, \by, \bbeta, \bw, \mathcal{G}, \blambda, \mfP) +
        \lambda_1 \sum_{i=1}^{|\mfP|} \sum_{k \in P_i} w_k |\beta_k| +  
        \lambda_2 \sum_{i=1}^{|\mfP|} \sum_{(k,l) \in E, k,l \in P_i, k < l} 
        w_{k,l}|\beta_{k} - \beta_l|.
\end{align*}
where we collect everything in the function $h(\bX, \by, \bbeta, \bw, \mathcal{G}, \blambda, \mfP)$
except for the effect of the $\lambda_1$-penalty and the 
$\lambda_2$-penalty inside the groups $P_i$. Then by the definition of the $P_i$, we see that 
$h$ is differentiable.  

We now want to find out, if we should split the set $P_i$, where we assume for simplicity
first that $\beta_k \neq 0$ for $k \in P_i$. As a mental image, we imagine
the elements in $P_i$ as balls. They are connected as defined by the edges in $\mathcal{G}$
with strings, which have tensile strength $\lambda_2 w_{kl}$. On each ball, we now
apply a force of $\partial h / \partial \beta_k + \lambda_1 w_k \text{sign}(\beta_k)$ for $k \in P_i$, where the sign 
determines if the force is applied upwards or downwards. The question
we want to answer is if some of the strings will break and if they break, which
of the balls will still be connected to each other. If the balls break apart
into separate groups, then we should also split $P_i$ into the corresponding groups.

We can specify this mental image in the form of a maximum flow problem. For this
restrict graph $\mathcal{G}$ to the set $P_i$ and add a source node $s$ (pulling upwards) and a sink node $r$ (pulling downwards). 
Our new graph is then $\mathcal{G}_{i}^M=(V_i^M, E_i^M)$, where 
$V_{i}^M = P_i \cup \{r, s\}$. For the edges, we define
\begin{align*}
E_i^M =& \{(k,l)|k,l \in P_i \text{ and } (k,l) \in E\} \cup \\
       & \{(k,s),(s,k) | \partial h / \partial \beta_k + \lambda_1 w_k \text{sign}
	(\beta_k) < 0\} \cup \\
       & \{(k,r),(r,k) | \partial h / \partial \beta_k + \lambda_1 w_k \text{sign}
	(\beta_k) > 0\}
\end{align*}
such that we keep all the edges in $\mathcal{G}$ that lie in $P_i$ and
add an edge to the source if we pull upwards on the coefficient and
an edge to the sink if we pull downwards. For the maximum flow problem to 
be complete, we now only have to specify the capacities on each of the nodes. These
correspond to the tensile strength of the cables in our little mental image
for edges between nodes in $P_i$ or the strength at which we pull on a node
for edges to either the source or the sink. More specifically, let
$c_{kl}$ be the capacity between nodes $k$ and $l$. Then
\begin{align*}
c_{kl} &= c_{lk} = \lambda_2 w_{kl} \text{ for } (k,l) \in E_i^M \text{ and } k,l \in P_i,\\
c_{sk} &= -\left(\frac{\partial h}{\partial \beta_k}+\lambda_1 w_k \text{sign}
  (\beta_k) \right) \text{ and } c_{ks} = 0 \text{ for } (k,s) \in E_i^M,\\
    \intertext{and}
c_{kr} &= \left(\frac{\partial h}{\partial \beta_k} + \lambda_1 w_k \text{sign}		
  (\beta_k)\right) \text{ and } c_{rk} = 0 \text{ for } (k,r) \in E_i^M.
\end{align*}
For this graph $\mathcal{G}_i^M$ we find the maximum possible flows and then
look at the residual graph $\mathcal{G}_i^R$. 
In our mental image,
the maximum flow corresponds to the maximum force we can exert before some strings
start breaking. If some edges from the source still have residual capacity,
this means that we could pull harder upwards, therefore breaking off some nodes -
exactly the ones that are still connected to the source in the residual graph.
The same is true for available residual capacity to the sink. Therefore we define
\begin{align*}
F_{i+} &= \{k \in P_i | k \text{ is connected to } s \text{ in }
 \mathcal{G}_{i}^R\}\\
F_{i-} &= \{k \in P_i | k \text{ is connected to } r \text{ in }
 \mathcal{G}_{i}^R\} \\
F_{i0} &= P_i \backslash (F_{i+} \cup F_{i-})
\end{align*}
These sets are not necessarily connected, in which case we separate them into their 
connected components. These are then the new fused sets, based on which we define
the fused matrices and vectors and run the component-wise algorithm. 

For the sets $P_i$ for which $\beta_k = 0$ for $k \in P_i$, things are just a 
little more complicated, as we also have to overcome the $\lambda_1$-penalty. 
This means that we have to define two flow graphs, one for trying to get $\beta_k$
to be positive, one for making $\beta_k$ negative. The first
graph is $\mathcal{G}_{i+}^M$, which is defined analogous to $\mathcal{G}_i^M$
above, except that instead of 
$\partial h / \partial \beta_k + \lambda_1 w_k \text{sign}(\beta_k)$ we use
$\partial h / \partial \beta_k + \lambda_1 w_k$, i.e. assume that 
$\text{sign}(\beta_k) = 1$. Correspondingly, we also define $\mathcal{G}_{i-}^M$
using $\partial h / \partial \beta_k - \lambda_1 w_k$, i.e. under the 
assumption that $\text{sign}(\beta_k) = -1$.
With $\mathcal{G}_{i+}^M$ and the correspond residual graph $\mathcal{G}_{i+}^R$,
we can now determine if we can break of a set that will become positive and using
$\mathcal{G}_{i-}^R$ if a set can become negative. These are then defined as
\begin{align*}
F_{i+} &= \{k \in P_i | k \text{ is connected to } s \text{ in }
 \mathcal{G}_{i+}^R\}\\
F_{i-} &= \{k \in P_i | k \text{ is connected to } r \text{ in }
 \mathcal{G}_{i-}^R\} \\
F_{i0} &= P_i \backslash (F_{i+} \cup F_{i-}).
\end{align*}
This now specifies how to generate the fused sets $\mfF$ using $\bbeta$ and $\mfP$.
In the next section, we will now present the complete algorithm and a proof that it
converges to the global optimum can be found in the appendix.

\subsection{The complete algorithm \label{Sec:CompleteAlgorithm}}
We now have specified everything that we need in order to run the algorithm.
We know how to define the sets $\mfP$ and $\mfF$ as well as given the grouping
of the coefficients how to calculate $\tilde{X}$, $\tilde{y}$, 
$\tilde{w}$ and $\tilde{\mathcal{G}}$ in order to apply the
component-wise algorithm to $\tilde{g}=g(\tilde{\bX}, \tilde{\by}, \tilde{\bbeta},
\tilde{\bw}, \tilde{\mathcal{G}}, \blambda)$. Then, given the solution of the
component-wise run, we specify a new grouping. We repeat
these steps until convergence. 

As the application of the maximum flow algorithm is relatively expensive, we add
one more rule in order to apply it as little as possible.
\begin{description}
\item[Fuse sets] First, we do not check any sets if they have to broken up and just set
    $F_i = P_i, \; i=1, \ldots, |\mfP|$. Run the component-wise algorithm.
\item[Split active sets] If we previously used the \emph{Fuse sets} rule 
    and $\bbeta$ did not change, then check sets $P_i$ with 
    $\beta_k \neq 0 \text{ for } k\in P_i$
    if they should be split up. Run the component-wise algorithm.
\item[Split inactive sets] If we previously used the \emph{Split active sets} rule
and $\bbeta$ did not change, then check all sets with $\beta_k = 0 \text{ for } k \in P_i$. Run the component-wise algorithm.
\end{description}
Using this we can avoid to calculate maximum flows until we have exhausted
the much easier and computationally cheaper fusion of sets. The complete
procedure is given in more detail in Algorithm \ref{Alg:CompleteCoordinate}.

\begin{algorithm}
    \KwData{$\by$, $\bX$, $\bbeta_0$}
    \KwResult{$\bbeta$ optimizing $g(\bX, \by, \bbeta, \bw, \mathcal{G}, \blambda)$}

    $\bbeta := \bbeta_0$\;
    $F_i := \{i\} \; i=1, \ldots, p$\;
    \Repeat{$F_i; \; i=1, \ldots, |F|$ did not change}{
        Based on $F_i; i=1, \ldots, |F|$ calculate $\tilde{\bbeta}$, $\tilde{X}$, $\tilde{y}$, 
        $\tilde{w}$ and $\tilde{\mathcal{G}}$\;
          
        Apply naive coordinate-wise to $\tilde{g}=g(\tilde{\bX}, \tilde{\by}, \tilde{\bbeta},
\tilde{\bw}, \tilde{\mathcal{G}}, \blambda)$

        Determine the new sets $F_i$ as described above according to \emph{Fuse sets},
        \emph{Split active sets} and \emph{Split inactive sets} rule\;
    }
    \caption{Coordinate-wise optimization algorithm for the Fused Lasso with grouping step.}
    \label{Alg:CompleteCoordinate}
\end{algorithm}

Now we can prove that this algorithm is guaranteed to converge to the global optimum.
\begin{thm}
Assume we want to minimize expression \ref{Eq:FusedLasso}. Then Algorithm 
\ref{Alg:CompleteCoordinate} is guaranteed to converge to the global optimum.
\label{Thm:CompleteCoordinate}
\end{thm}
The proof of this theorem can be found in the appendix.

\subsection{\label{Sec:Huber} Alternative algorithm using Huber penalty}
The algorithm presented above relies for convergence on solving maximum-flow problems.
When these are very large, they can become the bottleneck of the computation and slow it down.
Therefore, we propose a second algorithm that works without any maximum-flow computations
at the price of only yielding approximate solutions. 

The main part of the algorithm stays the same. We only replace the \emph{Split active sets}
and \emph{Split inactive sets} rule presented in Section \ref{Sec:CompleteAlgorithm} by an
application of the same problem using a Huber penalty with some $M > 0$, defined as
\[
p_M (x) = \begin{cases}
            \frac{M}{2} x^2 & \text{ for } -1/M \le x \le 1/M \\
            |x|- \frac{1}{2M} & \text{ otherwise}
          \end{cases}
\]
instead of an $L_1$ penalty on the differences of coefficients. It is easy to verify that $p_M(x)$ is continuous and differentiable everywhere
and that $p_M(x) \rightarrow |X|$ as $M \rightarrow \infty$. The loss function that we aim to optimize then is
\begin{align}
g_M(\bX, \by, \bbeta, \bw, \mathcal{G}, \blambda) =& 
    \frac{1}{2} (\by -\bX\bbeta)^T(\by - \bX\bbeta) + \nonumber \\
    &\lambda_1 \sum_{k=1}^p w_k|\beta_k| + \lambda_2 
    \sum_{(k,l) \in E, k < l} w_{k,l} p_M(\beta_{k} - \beta_l).
    \label{Eq:HuberLasso}
\end{align}
As for $g_M$ the penalty on differences $\beta_k - \beta_k$ is now differentiable everywhere, it can be shown that
a naive coordinate-wise optimization algorithm converges to a global optimum \citep[see][]{Tseng2001} and
that for $M \rightarrow \infty$, this global optimum converges to the global optimum of the Fused Lasso problem in
Equation (\ref{Eq:FusedLasso}). However, we will not use the approximation as given above directly to solve our
problem. The reason is that the convergence using the approximation in Equation (\ref{Eq:HuberLasso})
can be very slow when using a coordinate-wise algorithm. However, as it is differentiable everywhere,
the coordinate-wise algorithm cannot get stuck. Therefore, we will use this approximation if the 
naive algorithm presented in Section \ref{Alg:NaiveCoordinate} gets stuck in a coordinate-wise
optimum that is not the global optimum. 

To be more precise, for our algorithm we first use the naive coordinate-wise version 
presented in Algorithm \ref{Alg:NaiveCoordinate}. 
After each convergence, we \emph{fuse} all coefficients that are equal to each other and connected (i.e. the fused
sets are just $\mfP$ as defined in Section \ref{Sec:DefineFusedSets}) and reapply the naive coordinate-wise
procedure until nothing changes anymore. At this point, we may be stuck at a coordinate-wise minimum that
is not the global optimum. In order to get ``unstuck'', we apply the naive coordinate-wise procedure to the
loss function with Huber penalty $g_M$ for some value of $M$ for a number of iterations $K$. The precise value
of $M$ is not very important and we set it to 1000. As we already mentioned above, convergence of $g_M$ can take very long
so that we stop it after $K$ iterations (we set K=100), as this is usually sufficient to become ``unstuck''.
We iterate this procedure until $\bbeta$ has not changed more than $\varepsilon >0 $ according to some metric (e.g. $L_1$)
since the last application of the huberized loss function $g_M$. 

The algorithm as presented above is not guaranteed to converge, however as we will see in the simulations
section, it usually comes very close to the global optimum. The complete description can be seen in 
Algorithm \ref{Alg:CompleteCoordinateHuber}.

\begin{algorithm}
    \KwData{$\by$, $\bX$, $\bbeta_0$}
    \KwResult{$\bbeta$ optimizing $g(\bX, \by, \bbeta, \bw, \mathcal{G}, \blambda)$}

    $\bbeta := \bbeta_0$\;
    $\bbeta_{save} := \bbeta_0$\;
    $F_i := \{i\} \; i=1, \ldots, p$\;
    \Repeat{$||\bbeta - \bbeta_{save}||_1 < \varepsilon$}{
        Based on $F_i; i=1, \ldots, |F|$ calculate $\tilde{\bbeta}$, $\tilde{X}$, $\tilde{y}$, 
        $\tilde{w}$ and $\tilde{\mathcal{G}}$\;
          
        Apply naive coordinate-wise alg. to $\tilde{g}=g(\tilde{\bX}, \tilde{\by}, \tilde{\bbeta},
\tilde{\bw}, \tilde{\mathcal{G}}, \blambda)$\;

        Set $\mfF = \mfP$\;
        \If{$\mfF$ did not change}{
            Apply naive coordinate-wise alg. for at most $K$ steps to $g_M(\bX, \by, \bbeta, \bw, \mathcal{G}, \blambda)$\;
            Set $\bbeta_{save} := \bbeta$ \;
        }
        
    }
    \caption{Coordinate-wise optimization algorithm for the Fused Lasso employing Huber penalty.}
    \label{Alg:CompleteCoordinateHuber}
\end{algorithm}

All of these algorithms are not restricted to using squared error loss and therefore we will
quickly discuss some extensions to logistic regression aw well as the Cox proportional
hazards model.

\subsection{Extensions to logistic regression and the Cox proportional hazards model}
Of course, the Fused Lasso model is not restricted to squared error loss. 
Just as with the Lasso or the elastic net, we can apply the same mechanism also to
generalized linear models. By using iteratively reweighted least squares (IRWLS),
we can then leverage our algorithm for solving the Fused Lasso with squared error loss
to generalized linear models. Here, we will briefly describe the implementation for
logistic regression \citep[see also][]{FHT2010} and the Cox proportional hazards
model.

\subsubsection{Logistic regression}
In the setting of a binary response variable $\bY$ (with values $0$ and $1$),
logistic regression can be used. Here, the conditional probability of success 
conditioned on predictor vector $\bx$ is usually expressed as 
\[
P(Y=1|\bx) = \frac{1}{1+\exp(-\bx^T\bbeta)}   
\]
together with the Fused Lasso penalties, the function that we want to minimize is
\[
-\sum_{k=1}^n \log(P(Y=y_k|\bx_k)) + \lambda_1 \sum_{k=1}^p w_k|\beta_k| + 
    \lambda_2 \sum_{(k,l)\in E; k<l} w_{kl}|\beta_k - \beta_l|.
\]
In order to do this, we iteratively apply a quadratic approximation to the log-likelihood
$-\sum_{i=1}^n \log(P(Y=y_i|\bx_i))$ at the current estimates in step $m$, $\bbeta^{(m)}$. 
The exact details of this derivation can be looked up in many textbooks on generalized linear
models \citep[see e.g.][]{GeneralizedLinearModels} and therefore we omit them here. 
The resulting quadratic approximation has then the form
\[
\frac{1}{2} \sum_{k=1}^n v_k (z_k - \bx_k^T\bbeta)^2 + \lambda_1 \sum_{k=1}^p w_k|\beta_k| + \lambda_2 \sum_{(k,l)\in E; k<l} w_{kl}|\beta_k - \beta_l|
\]
where the working response $\bz$ and the weights $\bv$ are given by
\begin{align*}
z_k &= \bx_k^T \bbeta^{(m)} + \frac{y_k - p(\bx_k,\bbeta^{(m)})}
    {p(\bx_k,\bbeta^{(m)})(1 - p(\bx_k,\bbeta^{(m)}))}\\
v_k &= p(\bx_k,\bbeta^{(m)})(1 - p(\bx_k,\bbeta^{(m)})).
\end{align*}

\subsubsection{Cox proportional hazards model}
Another important regression model that is often used in many applications is the Cox 
proportional hazards model. For simplicity, we assume that there are no ties
between the event times $t_i$. Then the partial likelihood function is given by
\[
L(\bbeta) = \prod_{k=1}^n \left(\frac{\exp(\bx_k^T\bbeta)}{\sum_{l \in R(t_k)}
\exp(\bx_l^T \bbeta)} \right)^{\delta_k}
\]
where $t_k$ are the event times, $\delta_k = 0$ if the observation is censored and
$\delta_k=1$ otherwise, and $R(t_k) = \{l : t_l \ge t_k\}$ are the sets of
individuals at risk at time $t_k$.

As in the case of the Logistic regression model, we also want to approximate the
Cox model by a quadratic function around the current estimate of $\bbeta^{(m)}$. 
If we let 
\[
\bd = \frac{\partial \log L(\bbeta^{(m)})}{\partial \bbeta} \quad \text{and} \quad 
\bQ = \frac{\partial^2 \log L(\bbeta^{(m)})}{(\partial \bbeta)^2}
\]
then the quadratic approximation to $-\log L$ is
\[
-\log L(\bbeta^{(m)}) - \bd^T(\bbeta - \bbeta^{(m)}) - (\bbeta - \bbeta^{(m)})^T
    Q (\bbeta - \bbeta^{(m)}).
\]
The exact form of $\bd$ and $\bQ$ is easy to evaluate and can be 
found in many standard
textbooks \citep[see e.g.][]{GeneralizedAdditiveModels}. Specifically they are
\begin{align*}
\bd &= \bdelta^T \bX - \sum_{k=1}^n \delta_k \frac{\sum_{l \in R(t_k)} \exp(\bx_l^T \bbeta^{(m)} )}
{\sum_{l\in R_(t_k)} \exp(\bx_l^T \bbeta^{(m)})} \\
\bQ &= \bX^T \bW \bX \quad \text{with} \\
w_{kk} &= -\exp(\bx_k^T \bbeta^{(m)}) \sum_{l:k\in R_l} 
\left(\frac{1}{\sum_{j \in R_l} \exp(\bx_j^T \bbeta^{(m)})} + 
    \frac{\exp(2\bx_k^T \bbeta)}{\left(\sum_{j \in R_l} \exp(x_j^T \bbeta^{(m)}\right)^2}\right)\\
w_{kk'} &=  -\exp(\bx_k^T \bbeta^{(m)}) \exp(\bx_{k'}^T \bbeta^{(m)}) 
\sum_{l:k,k' \in R_l}
\frac{1}{\left(\sum_{j\in R_l}\exp(\bx_j^T\bbeta^{(m)})\right)^2}.
\end{align*}
The main problem is that $\bW$ is not a diagonal matrix, which greatly increases
the computational complexity of the calculation. \citet{GeneralizedAdditiveModels} suggest
to just use the diagonal of matrix $\bW$ for the computations, which we will do. However,as we are not optimizing the step size of the algorithm, this approach may lead to convergence problems, especially in the case of highly correlated predictors. As an ad-hoc solution to this problem, we use an additional diagonal matrix $\bD$ such that our approximate Hessian $\tilde{\bQ}$ is then
\[
\tilde{\bQ} = \bD \bX^T \text{Diag}{\bW} \bX \bD
\]
where $\bD$ is chosen such that $\text{Diag}(\tilde{\bQ}) = \text{Diag}
(\bQ)$. While we do not have a rigorous theoretical justification, we found that in 
practice this gave good convergence of the algorithm, while only requiring the
computation of the 
diagonal of $\bQ$, which is computationally still feasible.

\section{\label{Sec:Simulations}Simulations}
In this section, we want to demonstrate the speed of our new algorithm. For this,
we will simulate data and compare the speed of the naive, the maximum-flow based
and the Huber penalty based Fused Lasso algorithms to other methods discussed below. 
Additionally, as the naive and Huber penalty based methods are only approximate, 
we will also evaluate their performance in terms of accuracy with respect to the
exact Fused Lasso solution.

\subsection{Methods used}
To the best of our knowledge, there is currently no 
specialized software available that solves
the Fused Lasso for general matrices $\bX$ and general graphs $\mathcal{G}$.
Therefore we will use the \texttt{CVX}, a package for specifying and solving convex 
programs in Matlab \citep{cvx, gb08}. It is very versatile and is in our 
opinion the type of off-the shelf software that would be used to solve the Fused Lasso
in absence of other specialized software.

As the current algorithm is an extension of a coordinate-wise algorithm, we
will also compare the speed of our method to an implementation of this approach,
namely the package \texttt{glmnet} from \citet{FHT2007}. The \texttt{glmnet} package
implements a coordinate-wise descent algorithm to solve the elastic net, 
a combination of the lasso and ridge regression \citep[see][]{Zou2005}. It minimizes the
loss functions
\[
\frac{1}{2}(\by -\bX\bbeta)^T(\by - \bX\bbeta) + \lambda \alpha \sum_{k=1}^p w_k |\beta_k| + \lambda (1-\alpha) 
    \sum_{k=1}^{p-1} w_k \beta_k^2.
\]
Due to the simpler structure of the loss function that is being optimized
and the therefore simpler algorithm, we expect the \texttt{glmnet} package
to be faster than our implementation.

\subsection{\label{Sec:NormalData}Continuous response data example} 
For the simulation we have to generate a predictor matrix $\bX$ and 
a vector of true coefficients $\bbeta$. 

Here we distinguish if we want to generate data for simulations for a one- or
a two-dimensional graph. 

\paragraph{One-dimensional data}
For the predictor matrix, all observations are being generated independently
of each other. Let $p$ be the number of columns in the predictor matrix.
In order to give the matrix some structure, we set the matrix to a value
different from 0 for some intervals before adding Gaussian noise to it.
For observation $i$ let $n_i \sim \text{Poisson}(\sqrt{p}/2)$ be the number of 
such intervals, which have length $l_{ij} \sim \text{Poisson}(\sqrt{p})$, starting position
$s_{ij} \sim U(2-l_{ij}, p)$ and value
$v_{ij} \sim U(\{-3,-2,-1,0,1,2,3\})$, all independent for $i=1, \ldots, n; j=1, \ldots, n_i$.
For this we then set
\[
X_{ik} = v_{ij} + \gamma_{ik} \text{ for } s_{ij} \le k \le s_{ij} + l_{ij} -1; \quad i=1,\ldots, n; j=1, \ldots, n_i
\]
where $\gamma_{ik} \sim N(0,1) \; i.i.d.$ and in case of multiple assignment the one with the highest index $j$ takes precedence.

For the vector of coefficients, we choose
such that $\bbeta$ is 0 except for a sequence of length 100 in the middle.

Our response is then as usual
\[
y_i = (\bX\bbeta)_i + \varepsilon_i 
\]
where $\varepsilon_i \sim N(0, \sigma)$ for $\sigma = 10$.

\paragraph{Two-dimensional data}
We generate the two-dimensional data in a similar way. For ease of notation,
we assume that the predictor matrix is a 3-dimensional array (and cast it back into
a 2-dimensional matrix when we specify $\by$). Here let
$\bX \in \mathds{R}^{n \times p^2}$ and for observation $i$ we have $n_i \sim \text{Pois}(\sqrt{p})$
again the number of boxes. The boxes have a random length on both axis which is drawn from
$l_{ij}^{(1)},l_{ij}^{(2)} \sim \text{Pois}(\sqrt{p})$ and a uniform starting position on the axis 
using $s_{ij}^{(1)} \sim U(2-l_{ij}^{(1)}, p)$ and $s_{ij}^{(2)} \sim U(2-l_{ij}^{(2)}, p)$. The value of the 
box we again draw from $v_{ij} \sim U(\{-3,-2,-1,0,1,2,3\})$ where all these random variables
are all independent for $i=1, \ldots, n; j=1, \ldots, n_i$. Using all this we set
\[
X_{ik^{(1)}k^{(2)}} = v_{ij} + \gamma_{ik^{(1)}k^{(2)}} \text{ for } s_{ij}^{(1,2)} \le k^{(1,2)} \le s_{ij}^{(1,2)} + l_{ij}^{(1,2)} -1; \quad i=1,\ldots, n; j=1, \ldots, n_i
\]
where again $\gamma_{ik^{(1)}k^{(2)}} \sim N(0,1) \; i.i.d.$ and highest index $j$ takes precedence.

We also interpret the vector of coefficients $\bbeta$ as two-dimensional and pick a square in
the middle of length 10, i.e.
\[
\beta_{k^1k^2} = \begin{cases}
1 & \text{ if } p/2-4 \le k^{1,2} \le p/2+5 \\
0 & \text{ otherwise }
\end{cases}
\]
and the response is as before
\[
y_i = (\bX\bbeta)_i + \varepsilon_i 
\]
where $\varepsilon_i \sim N(0, \sigma)$ for $\sigma = 10$.

\subsection{Speed}
We now want to compare the speed of the algorithm and its approximations with \texttt{CVX} and \texttt{glmnet}.
In this comparison, it is important to note that \texttt{glmnet} solves a simpler problem than the 
Fused Lasso, but we include it as a faster benchmark.

Here, we run the algorithm for 20 values of $\lambda_2$ that span the relevant range from $[\lambda_2^{max}/10000, \lambda_2^{max}]$
using an exponential grid (i.e. there are about 5 values for $\lambda_2$ for each order of magnitude). 
Here $\lambda_2^{max}$ is the value for $\lambda_2$ for which when $\lambda_1=0$ all coefficients have
the same value. For each of the $\lambda_2$-values, we compute the solution for 50 values of $\lambda_1$ ranging in an 
exponential grid from $[\lambda_1^{max}/10000, \lambda_1^{max}]$ (i.e. $\approx 12$ values for every factor 10), 
where $\lambda_1^{max}$ is defined equivalently to $\lambda_2^{max}$,
guaranteeing that we cover the entire relevant range of values. In all these algorithms, we stop the computation
when more than $2 \cdot n$ elements of $\bbeta$ are non-zero as these results would be considered especially unreliable
and computation can become very slow with very many active variables.

The results can be seen in Tables \ref{Table:NormalOneDimSimSpeed} and \ref{Table:NormalTwoDimSimSpeed}. Any version
of the Fused Lasso algorithm is orders of magnitude faster than the results of the \texttt{CVX} package,
thereby making it possible to run the Fused Lasso for larger problems. Also, as expected, the 
\texttt{glmnet} package is faster than any of the Fused Lasso methods. This is easily explained as the underlying
mathematical problem that the \texttt{glmnet} procedure solves is considerably simpler 
than the Fused Lasso model. As discussed before, we included it as a faster benchmark. 
For the different Fused Lasso procedures we can see that naive implementation
is up to an order of magnitude faster than the exact procedure that is based on the maximum-flow algorithms.
However, as we will see later the accuracy of the naive method can be quite poor. The Huber penalty based
approach is usually as fast as the maximum-likelihood based method, but has speed advantages in cases
where $p >> n$. This can also be seen in Table \ref{Table:Largep}.

\begin{table}
\centering
\begin{tabular}{|c|c|c|c|c|c|c|} \hline
\multicolumn{2}{|c|}{Param.} & \multicolumn{2}{|c|}{Other} & \multicolumn{3}{|c|}{FusedLasso} \\ \hline
n & p & Glmnet & CVX & exact & naive & Huber \\ \hline
100 & 1000 & 1.14 & 5883 & 19.39 & 3.32 & 16.90\\
100 & 3000 & 3.71 & 21363 & 87.73 & 10.11 & 50.67 \\
100 & 5000 & 6.90 & 40691 & 196.56 & 15.67 & 90.82 \\
200 & 1000 & 2.09 & 16582 & 37.99 & 7.47 &  39.08 \\
200 & 3000 & 6.72 & 65328 & 132.66 & 20.16 & 94.57\\
200 & 5000 & 12.15 & 122197 & 277.13 & 32.17 & 176.19  \\
600 & 1000 & 4.02 & 107820 & 105.40 & 24.74 & 144.98  \\
600 & 3000 & 18.37 & - & 345.16 & 58.01 & 331.57 \\
600 & 5000 & 27.27 & - & 576.41 & 82.46 & 482.28 \\
1000 & 1000 & 73.80 & 303187 & 204.34 & 49.60 & 321.34 \\
1000 & 3000 & 27.59 & - & 574.31 & 98.23 & 618.09\\
1000 & 5000 & 51.17 & - & 1042.83 & 163.51 & 1007.66 \\ \hline
\end{tabular}
\caption{Speed of the algorithms for continuous data with one-dimensional graph in seconds.}
\label{Table:NormalOneDimSimSpeed}
\end{table}

\begin{table}
\centering
\begin{tabular}{|c|c|c|c|c|c|c|} \hline
\multicolumn{2}{|c|}{Param.} & \multicolumn{2}{|c|}{Other} & \multicolumn{3}{|c|}{FusedLasso} \\ \hline
n & p & Glmnet & CVX & exact & naive & Huber \\ \hline
100 & 30x30 & 1.55 & 19867 & 63.05 & 13.10 & 34.59\\
100 & 50x50 & 5.93 & 29058 & 264.62 & 26.36 & 97.06\\
100 & 70x70 & 12.32 & 64653 & 780.88 & 37.48 & 160.04\\ 
200 & 30x30 & 3.07 & 42311 & 126.37 & 30.22 & 83.87\\
200 & 50x50 & 10.18 & 67770 & 374.12 & 52.49 & 173.29\\
200 & 70x70 & 19.32 & 165690 & 933.70 & 62.30 & 267.12\\ 
600 & 30x30 & 4.67 & 134887 & 173.62 & 43.36 & 152.36\\
600 & 50x50 & 25.51 & - & 728.01 & 125.85 & 493.39\\ 
600 & 70x70 & 35.94 & - & 1584.98 & 166.30 & 686.73\\ 
1000 & 30x30 & 51.61 & 275220 & 209.46 & 53.16 & 195.20\\
1000 & 50x50 & 36.43 & - & 1179.99 & 214.05 & 848.18\\ 
1000 & 70x70 & 36.03 & - & 1807.08 & 221.78 & 895.51\\ \hline
\end{tabular}
\caption{Speed of the algorithms for continuous data with two-dimensional graph in seconds.}
\label{Table:NormalTwoDimSimSpeed}
\end{table}

\begin{table}
\centering
\begin{tabular}{|c|c|c|c|c|c|} \hline
\multicolumn{2}{|c|}{Param.} & \multicolumn{1}{|c|}{Other} & \multicolumn{3}{|c|}{FusedLasso} \\ \hline
n & p & Glmnet & exact & naive & Huber \\ \hline
100 & 10000 & 31.25 & 753.5 & 65.87 & 291.633\\
100 & 20000 & 60.62 & 2509.4 & 115.89 & 515.66\\
100 & 100x100 & 24.21 & 3161.8 & 50.6 & 289.1\\
100 & 150x150 & 46.55 & 22164.49 & 86.6 & 502.1 \\ \hline
\end{tabular}
\caption{Speed of the algorithms for continuous data with one- and two-dimensional graph for
large number of variables in seconds.}
\label{Table:Largep}
\end{table}

\subsection{Accuracy}
In Section \ref{Sec:Algorithm} we already mentioned that the naive algorithm as well as 
the Huber penalty based version are not guaranteed to converge to the global optimum. Therefore,
we want to assess how far away from the global optimum, that we get with the exact 
maximum-flow based algorithm, the solutions are. As the solutions and therefore also the
error are dependent on the choice of the penalty parameter, we report the worst-case
scenario over the whole grid of $\lambda_1$ and $\lambda_2$ values. If we let
$\hat{\bbeta}^{\text{exact}}(\lambda_1, \lambda_2)$, $\hat{\bbeta}^{\text{Huber}}(\lambda_1, \lambda_2)$ and
$\hat{\bbeta}^{\text{naive}}(\lambda_1, \lambda_2)$ be the solution for penalty parameters
$(\lambda_1, \lambda_2)$ for the maximum-flow algorithm (exact solution), Huber penalty approach
and naive algorithm respectively, then we define the reported accuracy as
\[
\max_{\lambda_1, \lambda_2} err\left(\hat{\bbeta}^{\text{exact}}(\lambda_1, \lambda_2) - \hat{\bbeta}^{\text{Huber}}(\lambda_1, \lambda_2)\right)
\]
and
\[
\max_{\lambda_1, \lambda_2} err\left(\hat{\bbeta}^{\text{exact}}(\lambda_1, \lambda_2) - \hat{\bbeta}^{\text{naive}}(\lambda_1, \lambda_2)\right)
\]
for the error functions
\begin{align*}
err(x) &= \frac{1}{p} \sum_{i=1}^p |x_i| = ||x||_1/p \\
err(x) &= \sqrt{\frac{1}{p} \sum_{i=1}^p x_i^2} = RMSE\\
err(x) &= \max_{i=1, \ldots, p} |x_i| = ||x||_\infty.
\end{align*}

The accuracies for the one-dimensional and two-dimensional simulations can be seen in Tables \ref{Table:NormalOneDimSimAccuracy} and 
\ref{Table:NormalTwoDimSimAccuracy}. As we can see there, the accuracy of the Huber penalty based approach is in general very good,
showing small errors for $||\cdot||_1/p$ and the RMSE. Looking at $||\cdot||_\infty$, we can see that even for every single component,
errors are usually below 0.1 (the non-zero elements in the true $\bbeta$-vector are set to 1). Overall, the Huber penalty based
approximation gives a good performance, while the error for the naive approach is usually an order of magnitude larger and
in some cases very high (see $||\cdot||_\infty$-measure).

\begin{table}
\centering
\begin{tabular}{|c|c|c|c|c|c|c|c|} \hline
\multicolumn{2}{|c|}{Param.} & \multicolumn{2}{|c|}{$||\cdot||_1/p$} & \multicolumn{2}{|c|}{RMSE} & \multicolumn{2}{|c|}{$||\cdot||_\infty$}\\ \hline
n & p & naive & Huber & naive & Huber & naive & Huber\\ \hline
100 & 1000 & 0.088 & 0.0077 & 0.19 & 0.016 & 1.36 & 0.16\\
100 & 3000 & 0.033 & 0.0040 & 0.11 & 0.0096 & 1.39 & 0.11\\
100 & 5000 & 0.018 & 0.0043 & 0.079 & 0.011 & 1.16 & 0.15\\
200 & 1000 & 0.077 & 0.0023 & 0.13 & 0.0044 & 0.80 & 0.044\\
200 & 3000 & 0.030 & 0.0025 & 0.083 & 0.0052 & 0.84 & 0.077\\
200 & 5000 & 0.0188 & 0.0028 & 0.057 & 0.0060 & 0.90 & 0.080\\
600 & 1000 & 0.16 & 0.030 & 0.27 & 0.034 & 1.17 & 0.25 \\
600 & 3000 & 0.023 & 0.0010 & 0.045 & 0.0022 & 0.38 & 0.026\\
600 & 5000 & 0.013 & 0.0014 & 0.032 & 0.0031 & 0.46 & 0.034 \\
1000 & 1000 & 0.16 & 0.024 & 0.27 & 0.032 & 1.05 & 0.30 \\
1000 & 3000 & 0.017 & 0.0017 & 0.046 & 0.0033 & 0.50 & 0.077 \\
1000 & 5000 & 0.008 & 0.0012 & 0.023 & 0.0024 & 0.37 & 0.026 \\ \hline
\end{tabular}
\caption{Accuracy of the algorithms for continuous data with one-dimensional graph. Accuracy
is measured as $||\cdot||_1/p$, root mean-squared error (RMSE) and $||\cdot||_\infty$, where for each of these measures,
the worst result over all $\lambda_1, \lambda_2$ combinations is being reported and averaged over the 10 simulations.}
\label{Table:NormalOneDimSimAccuracy}
\end{table}

\begin{table}
\centering
\begin{tabular}{|c|c|c|c|c|c|c|c|} \hline
\multicolumn{2}{|c|}{Param.} & \multicolumn{2}{|c|}{$||\cdot||_1/p$} & \multicolumn{2}{|c|}{RMSE} & \multicolumn{2}{|c|}{$||\cdot||_\infty$}\\ \hline
n & p & naive & Huber & naive & Huber & naive & Huber\\ \hline
100 & 30x30 & 0.052 & 0.016 & 0.13 & 0.045 & 0.83 & 0.35\\
100 & 50x50 & 0.015 & 0.0051 & 0.066 & 0.024 & 0.85 & 0.32\\
100 & 70x70 & 0.008 & 0.0021 & 0.049 & 0.013 & 0.89 & 0.25\\ 
200 & 30x30 & 0.048 & 0.0085 & 0.095 & 0.018 & 0.59 & 0.11\\
200 & 50x50 & 0.019 & 0.0035 & 0.062 & 0.012 & 0.64 & 0.13\\
200 & 70x70 & 0.009 & 0.0017 & 0.042 & 0.0083 & 0.54 & 0.13\\ 
600 & 30x30 & 0.18 & 0.0062 & 0.28 & 0.012 & 0.94 & 0.091\\
600 & 50x50 & 0.026 & 0.0017 & 0.048 & 0.0045 & 0.31 & 0.048\\ 
600 & 70x70 & 0.013 & 0.0013 & 0.034 & 0.0039 & 0.32 &  0.051\\ 
1000 & 30x30 & 0.17 & 0.0074 & 0.26 & 0.011 & 0.84 & 0.084\\
1000 & 50x50 & 0.023 & 0.0012 & 0.037 & 0.0024 & 0.26 & 0.035\\ 
1000 & 70x70 & 0.015 & 0.0011 & 0.029 & 0.0026 & 0.23 & 0.036\\ \hline
\end{tabular}
\caption{Accuracy of the algorithms for continuous data with two-dimensional graph. Accuracy
is measured as $||\cdot||_1/p$, root mean-squared error (RMSE) and $||\cdot||_\infty$, where for each of these measures,
the worst result over all $\lambda_1, \lambda_2$ combinations is being reported and averaged over the 10 simulations.}
\label{Table:NormalTwoDimSimAccuracy}
\end{table}

\section{\label{Sec:Discussion}Discussion}
In this article we have presented two novel algorithms based on coordinate-wise optimization 
that solve the Fused Lasso, both of which
are considerably faster than currently available methods. For the maximum-flow based approach
we have proven that it is guaranteed to converge to the global optimum, however, for problems
with large numbers of variables this can be slow due to the complexity of maximum-flow algorithms. 
In order to remedy this problem, we also introduced a Huber penalty based procedure that does not
rely on maximum-flow problems and shows better performance for situations with $p \gg n$.

Apart from this we also extended the Fused Lasso to allow for more general penalty structures
by using arbitrary undirected graphs and weights as well as more response types
by implementing logistic regression as well as the Cox proportional hazards model. 
An implementation of this algorithm will be provided in the \texttt{R}
package \texttt{FusedLasso} that will be available on \texttt{CRAN}.

\appendix
\section{A brief introduction to subgradients}
In convex optimization, often some of the functions in a problem are not 
differentiable everywhere. In this case, instead of a regular gradient, we can 
use a subgradient. The following introduction is mostly taken from 
\citet{Bertsekas1999}. As a starting point, we first want to define what a 
subgradient is and then describe a condition that guarantees that a convex 
function has a minimum at a point $x$. Afterwards, we will derive the relevant 
subgradient expressions that are being used in this article.
\begin{definition}
Given a convex function $f : \mathds{R}^n \rightarrow \mathds{R}$, we say a 
vector $d \in \mathds{R}^n$ is a \emph{subgradient} of $f$ at a point $x \in \mathds{R}^n$ if
$$
f(z) \ge f(x) + (z-x)'d, \quad \forall \, z \in \mathds{R}^n.
$$
If instead $f$ is a concave function, we say that $d$ is a subgradient 
of $f$ if $-d$ is a subgradient of $-f$ at $x$. The set of all subgradients of 
a convex function $f$ at $x \in \mathds{R}^n$ is called the sub-differential of 
$f$ at $x$, and is denoted by $\partial f(x)$.
\end{definition}
For subgradients, some basic properties similar to regular gradients hold and 
the proof to the following Proposition can be found in \citet{Bertsekas1999} on pp. 712-716.
\begin{prop}
Let $f : \mathds{R}^n \rightarrow \mathds{R}$ be convex. For every $x \in \mathds{R}^n$, the following hold:
\begin{enumerate}
\item If $f$ is equal to the sum $f_1 + \cdots + f_m$ of convex functions 
$f_j : \mathds{R}^n \rightarrow \mathds{R},
j=1, \ldots, m$, then $\partial f(x)$ is equal to the vector sum 
$\partial f_1(x) + \cdots + \partial f_m(x)$.
\item If $f$ is equal to the composition of a convex function $h : \mathds{R}^m 
\rightarrow \mathds{R}$ and an $m \times n $ matrix $\bA$, that is [$f(x) = h(\bA x)$], 
then $\partial f(x)$ is equal to $\bA' \partial h(\bA x) = \{\bA'g : g \in \partial h(\bA x)\}$.
\item $x$ minimizes $f$ over a convex set $\mathcal{A} \subset \mathds{R}^n$ if 
and only if there exists a subgradient $d \in \partial f(x)$ such that
$$
d'(z-x)\ge 0. \quad \forall z \in \mathcal{A}.
$$
\end{enumerate}
\end{prop}

In our case here for the convex functions we optimize over the set 
$\mathcal{A} = \mathds{R}^p$ and therefore the last statement says that $x$ 
minimizes a function $f$ over $\mathds{R}^p$ if and only if a subgradient 
$d \in \partial f(x)$ exists such that
$$
d = 0.
$$
Now, by using the proposition from above, all we need to calculate the 
subgradient of the loss function is the subgradient of $f (x) = |x|$ which is 
$\partial f (0) = [-1, 1]$ and $\partial f(x) = \text{sign}(x)$ for $x \neq 0$. 
Therefore, the subgradient of the loss function
$$
\frac{1}{2}(\by-\bX\bbeta)^{T}(\by-\bX\bbeta)+\lambda_{1}\sum_{k=1}^{p}w_k|\beta_{k}|+
\lambda_{2}\sum_{(k,l)\in E,k<l}w_{kl}|\beta_{k}-\beta_{l}|
$$
w.r.t. $\beta_k$ is
$$
-\bx^T_k(\by-\bX\bbeta) + \lambda_{1}\sum_{k=1}^{p} w_k\frac{\partial f(\beta_{k})}{\partial \beta_k} + 
\lambda_{2}\sum_{(k,l)\in E,k<l} w_{kl}\frac{\partial f(\beta_{k}-\beta_{l})}{\partial \beta_k}
$$
and using the optimality condition from above, a solution $\bbeta$ is optimal 
if there exists $s_k$, $t_{kl}$ such that 
$$
-\bx^T_k(\by-\bX\bbeta)_k + \lambda_{1} s_k + \lambda_{2}\sum_{(k,l)\in E,k<l} t_{kl} \quad \text{for} 
\quad k=1, \ldots, p
$$
where $s_k = \text{sign} (\beta_k)$ for $\beta_k \neq 0$ and $s_k \in [-1,1]$ 
otherwise. Similarly $t_{kl} = \text{sign}(\beta_k - \beta_l)$ for $\beta_k = \beta_l$ 
and $t_{kl} \in [-1,1]$ otherwise. For notational convenience, we set 
$t_{lk} = -t_{kl}$ for all $k < l$.

\section{Proof of Theorem \ref{Thm:NaiveCoordinate}}
Before we start with the main proof, we show the following lemma, where we
guarantee that the step size of each coordinate-move in the naive coordinate-wise
algorithm has to converge to 0.

\begin{lem}
Assume that in the coordinate-wise algorithm we are starting at point $\bbeta$
and are optimizing coordinate $k_0$. Let $\hat{\bbeta}$ be the new estimate
after the optimization of coordinate $k_0$, i.e. $\beta_k = \hat{\beta}_k$ for 
all $k \neq k_0$. Then there exists a constant $a$ such that
\[
g(\hat{\bbeta}) \le g(\bbeta) - \alpha(\beta_{k_0} - \hat{\beta}_{k_0})^2.
\]
The value of $a$ is independent of $k_0$ and the starting point $\bbeta$.
Here, we have suppressed the dependence of $g$ on other variables in the model
for notational convenience.
\label{Lem:MinimalImprovement}
\end{lem}
\begin{proof}
If we just consider function $g(\bbeta)$ as a function of $\beta_{k_0}$, 
then it has the form
\[
q(\beta_{k_0}) = a_{k_0}(\beta_{k_0} - b)^2 + \sum_{i=0}^K d_i |\beta_{k_0} - c_i|
\]
where $a_{k_0} = (X^TX)_{k_0k_0}$, $d_i >0$ and $b, c_i \in \mathds{R}$. The
function $q$ is differentiable everywhere except at point $c_i$ and has the
derivative
\[
q'(\beta_{k_0}) = 2a_{k_0}(\beta_{k_0} - b) - \sum_{i|\beta_{k_0} < c_i} d_i + 
\sum_{i | \beta_{k_0} > c_i} d_i
\]
which is clearly a piecewise linear function with slope $2a_{k_0}$ and jumps for
$\beta_{k_0} = c_i$ with height $2 d_i > 0$. For the minimum $\hat{\beta}_{k_0}$
we therefore know that
\[
\lim_{\beta_{k_0} \rightarrow \hat{\beta}_{k_0}-} q'(\beta_{k_0}) \le 0 \le
\lim_{\beta_{k_0} \rightarrow \hat{\beta}_{k_0}+} q'(\beta_{k_0})
\]
and as all jumps are positive it holds
\[
q'(\beta_{k_0}) \ge 2a_{k_0}(\beta_{k_0} - \hat{\beta}_{k_0}) \quad \text{for} \quad 
\beta_{k_0} > \hat{\beta}_{k_0}
\]
and
\[
q'(\beta_{k_0}) \le 2a_{k_0}(\beta_{k_0} - \hat{\beta}_{k_0}) \quad \text{for} \quad 
\beta_{k_0} < \hat{\beta}_{k_0}.
\]
However, from this it immediately follows that
\[
q(\beta_{k_0}) \ge q(\hat{\beta}_{k_0}) + a_{k_0} (\beta_{k_0} - \hat{\beta}_{k_0})^2.
\]
This claim still holds if instead of $a_{k_0}$ we use $a = \min_k a_k$.
As $g(\bbeta_{k_0}) = q(\bbeta_{k_0})$ by construction, the claim follows.
\end{proof}

We now go on to prove Theorem \ref{Thm:NaiveCoordinate}. 
The proof closely follows
parts of the proof of \citet[Theorem 4.1]{Tseng2001}.
\begin{proof}
Assume that our algorithm starts at $\bbeta^{(0)}$ and let $\bbeta^{r}$ be
the r-th step of the component-wise algorithm. Note that 
$B=\{\bbeta | g(\bbeta) \le g(\bbeta^{(0)}) \}$
is a compact set and as each coordinate-wise step only decreases $g$, 
it follows that
$\bbeta^{r} \in B$ for all $r$. 

Now let $r \in \mathcal{R}$ be a converging subsequence with
\[
\lim_{r \in \mathcal{R}} \bbeta^{r} = \bdelta.
\]
For each $j \in \{1, \ldots, p\}$, consider the subsequence $\bbeta^{r + j}$ for
$r \in \mathcal{R}$. 
As $g$ is continuous and $g(\bbeta^{r})$ is monotonically
decreasing, we get
\[
\lim_{r \rightarrow \infty} g(\bbeta^{r}) = g(\bdelta)
\]
exist. This, together with Lemma \ref{Lem:MinimalImprovement} 
then also implies that 
$\bbeta^{r+j} - \bbeta^{r+j+1} \rightarrow 0$ and thus
\[
\lim_{r \rightarrow \infty} \bbeta^{r+j} = \bdelta
\]
so that the limit of the sequence that is shifted by $j$ 
also exists and is equal to the limit of the unshifted sequence.

Now we want to show that then, $\bdelta$ is coordinate-wise minimum w.r.t. $g$. 
Let $e_k$ be the vector that is $1$ at position $k$ and $0$ otherwise. 
Furthermore let $k(r) = (r-1 \text{ mod } p) + 1$ 
be the coordinate index that is being
optimized in step $r$. Then we know that 
\[
g(\bbeta^{r+j} + \lambda e_{k(r+j)}) \ge g(\bbeta^{r+j}) \quad \forall \lambda, 
    \forall j \in \{1, \ldots, p\}.
\]
By using the subsequence that always moves coordinate $k$ and going to the limit we
then have
\[
g(\bdelta + \lambda e_k) \ge g(\bdelta)
\]
which holds for every $k \in \{1, \ldots, p\}$. Therefore, $\bdelta$ is guaranteed
to be a coordinate-wise optimum. 

Our proof as it is shown here works for an algorithm that always moves every coordinate
and does not leave any out. However, in Algorithm \ref{Alg:NaiveCoordinate}, we included
an active set $\mathcal{A}$. It can easily be seen that this is not a problem.
Our proof works as it is for all coefficients that are included in $\mathcal{A}$. 
Furthermore, $\mathcal{A}$ always adds coefficients, for which it is not
optimal to remain at $0$ but never excludes any. Therefore, 
after a finite number of steps in the outer iteration, all variables
that are non-zero at their optimum are included. This concludes our theorem. 
\end{proof}

\section{Proof of Theorem \ref{Thm:CompleteCoordinate}}
Before we go to the main proof, we have to show a few lemmata:
\begin{lem}
Let $\bbeta$ have associated partition $\mfP$ and transformed problem
$\tilde{g}$, for which $\tilde{\bbeta}$ is coordinate-wise optimal. Also,
let $\mfF$ be the corresponding fused sets. Assume that $\beta_k = 0$ for all
$k \in P_i$ for some $i$ and $P_i = F_i$, i.e., the set $P_i$ cannot be split (according
to the ``Split inactive set'' rule). Then there exist
$s_k \in [-1,1]$ and $t_{kl} \in [-1,1]$ for $k,l \in P_i$ for which
\begin{equation}
\partial h / \partial \beta_k  + \lambda_1 w_k s_k  + \lambda_2 \sum_{l \in E_i} w_{kl} t_{kl} = 0 
    \quad \forall k \in P_i
\label{Eq:FLSABeta0}
\end{equation}
where $E_i = \{(k,l) \in E: k,l \in P_i \}$.
\label{Lem:Beta0}
\end{lem}
\begin{proof}
In order to show this, 
assume that such $s_k$ and $t_{kl}$ do not exist. Now look at
the Fused Lasso problem
\[
\sum_{k \in P_i} \frac{1}{2\lambda_1 w_k}
(\partial h / \partial \beta_k  + \lambda_1 w_k s_k)^2 + 
\lambda_2 \sum_{(k,l) \in E_i, k < l} w_{kl}|s_k - s_l|.
\]
which has its optimum at, say, $s_k^0$ w.r.t. $\bs$. It has the subgradients w.r.t. $s_k$
\[
\partial h / \partial \beta_k  + \lambda_1 w_k s_k^0 + \lambda_2 \sum_{l \in E_i} w_{kl} t_{kl} = 0 
    \quad \forall k \in P_i
\]
and by optimality we know that the $t_{kl} \in [-1,1]$ with $t_{kl} = \text{sign}
(s_k -s_l)$ for $s_k \neq s_l$.   
By our assumption it follows that there exists a $k \in P_i$ with $s_k^0 \not \in [-1,1]$ and w.l.o.g. we assume that $\exists s_k^0 > 1$.
Let
\[
S = \{k | s_k^0 > 1 \} 
\]
Then we know that $t_{kl} = 1$ for all $k \in S$ and $l \not \in S$. Therefore,
we have 
\begin{align*}
0 =& \sum_{k \in S} \left(\partial h / \partial \beta_k  + \lambda_1 w_k s_k^0 +
 \sum_{l:l\in P_i, (k,l) \in E} \lambda_2 w_{kl} t_{kl}\right) > \\
& \sum_{k \in S} \left(\partial h / \partial \beta_k  + \lambda_1 w_k +
 \sum_{l:l\in P_i \backslash S, (k,l) \in E} \lambda_2 w_{kl} t_{kl}\right) 
\end{align*}
where in the inequality we use $s_k^0 > 1$ for $k \in S$ as well as $t_{kl} = -t_{lk}$.
Then
\[
-\sum_{k \in S} \partial h / \partial \beta_k  - \lambda_1 w_k >
 \sum_{k \in S}\sum_{l:l\in P_i \backslash S, (k,l) \in E} \lambda_2 w_{kl}  
\]
where we used that $t_{kl} = 1$ for all $k \in S, l \not \in S$.
In this equation, on the left hand side we have the sum over all capacities
coming out of the source into $S$ in $\mathcal{G}_{i+}^M$. On the right hand
side, we have the sum over all edges from $S$ into $P_i \backslash S$, which by the 
max-flow-min-cut theorem in graph theory is the maximal flow possible in
the graph coming out of nodes $S$. This implies, that the source is connected
to at least one node in $S$ in $\mathcal{G}_{i+}^R$, implying that $F_i \neq P_i$.
However this is a contradiction and therefore the claim holds.
\end{proof}

We also need the following lemma that characterizes $\bbeta$ where the 
associated partition $\mfP$ can be split and those where it cannot.
\begin{lem}
Assume that $\bbeta^{(m)}$ has partition $\mfP$ and associated transformed $\tilde{\bbeta}^{(m)}$
and target function $\tilde{g}$. Also assume that $\tilde{\bbeta}^{(m)}$ is 
coordinate-wise optimal w.r.t $\tilde{g}$. Furthermore, let $\mfF$ be
the fused sets associated with $\mfP$ and $\bbeta^{(m)}$. Then, the following statements hold
\begin{enumerate}[(a)]
\item If $\mfF \neq \mfP$, then $\bbeta^{(m+1)} \neq \bbeta^{(m)}$. \label{Enum:CaseSplitSet}
\item $\bbeta^{(m)}$ is the global optimum iff $\mfP = \mfF$. \label{Enum:Optimum}
\end{enumerate}
\label{Lem:ConvergenceConditions}
\end{lem}
\begin{proof}
Let us first deal with case (\ref{Enum:CaseSplitSet}). The condition implies, that 
at least one of the sets of $\mfP$ has been split. Let
$P_i$ be a set that has been split and $F_{i+} \neq \emptyset$ (otherwise, switch signs).
Then the condition for splitting a set implies 
that the coordinate-wise algorithm will move at least one of the coordinates. To see this, 
let $f_{kl}$ for $k,l \in P_i \cup \{r,s\}$ be the solution of the with $P_i$ and graph
$\mathcal{G}_{i+}^M$ associated maximum-flow problem. Then by definition of 
$F_{i+}$ we know that 
\begin{align*}
&\sum_{k \in F_{i+}} -\frac{\partial h}{\partial \beta_k} - \lambda_1 w_k = \sum_{k \in F_{i+}}
c_{sk}  > \\
& \sum_{k \in F_{i+}} f_{sk} = \sum_{(k,l): k \in F_{i+}, l\in P_i \backslash F_{i+},
    (k,l)\in E_{i+}^M} c_{kl}
\end{align*}
where the last inequality follows from the max-flow-min-cut theorem of graph theory \citep[Chapter 26]{Cormen2001}.
If $F_{i+}$ is not one connected group,
but consists of several disconnected components, this is true for every component.
Therefore, we can just assume that $F_{i+}$ consists of just one component.
By the definition of our constrained problem we know that
\[
g(\tilde{\bX}, \tilde{\by}, \tilde{\bbeta}^{(m)}, \tilde{\bw}, \tilde{\mathcal{G}}, \blambda) = 
g(\bX, \by, \bbeta^{(m)}, \bw, \mathcal{G}, \blambda)
\]
and 
\[
\frac{\partial_+ g(\tilde{\bX}, \tilde{\by}, \tilde{\bbeta}, \tilde{\bw}, 
        \tilde{\mathcal{G}}, \blambda)}{\partial \tilde{\beta}_i} = 
\sqrt{|F_i|} \frac{\partial_+ g(\bX, \by, \bbeta, \bw, \mathcal{G}, \blambda)}{\partial \bv}.
\]
where we take the directional derivative along the vector $\bv$ with 
$v_k = 1/\sqrt{|F_i|}$ for $k \in F_i$ and $v_k = 0$ for $k \not \in F_i$.
As $h$ by the definition of $P_i$ is differentiable at $\bbeta^{(m)}$ w.r.t. $\bv$, we get
\begin{align*}
\sqrt{|F_{i+}|} &\frac{\partial_+ g(\bX, \by, \bbeta^{(m)}, \bw, \mathcal{G}, \blambda)}{\partial \bv}=\\
&\sum_{k \in F_{i+}} \frac{\partial h(\bX, \by, \bbeta^{(m)}, \bw, \mathcal{G}, \blambda)}{\partial \beta_k} + 
\lambda_1 w_k + \lambda_2 \sum_{k \in F_{i+}, l \in P_i \backslash F_{i+}; (k,l) \in E} w_{kl} =\\ 
&\sum_{k \in F_{i+}} c_{sk} + 
\sum_{k \in F_{i+}, l \in P_i \backslash F_{i+}; (k,l) \in E} f_{kl} < 0 
\end{align*}
as deduced from the inequality above. Therefore, after applying the coordinate-wise optimization,
we have at least moved the component of  $\tilde{\beta}$ corresponding to $F_{i+}$ and get $\bbeta^{(m+1)} \neq \bbeta^{(m)}$.

Now we show (\ref{Enum:Optimum}). 
First, we assume that $\bbeta^{(m)}$ is indeed optimal. From the fact that the coordinate-wise
algorithm always improves $\bbeta^{(m)}$ w.r.t. $g$, part (\ref{Enum:CaseSplitSet}) proves
that
\[
\mfF \neq \mfP \quad \Longrightarrow \quad \bbeta^{(m)} \text{ not optimal}.
\]
Negating this just proves the first part of the claim.

For the opposite direction, assume that $\mfF = \mfP$.
In order to show optimality of a non-differentiable function, we can use subgradients again.
For our case, it says that a solution $\bbeta$ is optimal, if there are $s_k$ 
for $k \in 1, \ldots, p$ and $t_{kl}$ with $(k,l) \in E$ such that
\begin{equation}
-(\bX^T(\by - \bX \bbeta))_k + \lambda_1 w_k s_k + \lambda_2 \sum_{l : (k,l) \in E} w_{kl} t_{kl} = 0
\label{Eq:Subgradient}
\end{equation}
for all $k$. Here, it has to hold that $s_k = \text{sign}(\beta_k)$ if $\beta_k \neq 0$
and $s_k \in [-1,1]$ for $\beta_k = 0$. Similarly, $t_{kl} = \text{sign}(\beta_k-\beta_l)$
for $\beta_k \neq \beta_l$ and $t_{kl} \in [-1,1]$ for $\beta_k = \beta_l$.

As specified before in Section \ref{Sec:DefineFusedSets}, the sums over the
$t_{kl}$ can be divided into the $t_{kl}$ within sets $P_i$ and those between. Then
Equation (\ref{Eq:Subgradient}) is
\begin{equation}
\partial h / \partial \beta_k  + \lambda_1 w_k s_k + \lambda_2 \sum_{l\in P_i : (k,l) \in E} w_{kl} t_{kl} = 0.
\label{Eq:SubgradientWithGamma}
\end{equation}

First consider the case where $\beta_k \neq 0$ for $k \in P_i$. Here we have then that
$s_k = \text{sign}(\beta_k)$.
Then let $f_{kl}$ for $k,l \in P_i$ be the solution of graph $\mathcal{G}^M_{i}$.
From the assumption that $P_i=F_i$, we know that in the residual graph
$\mathcal{G}^R_{i}$, no node is connected to the source or the sink, all edges coming
from the source or going into the sink are at maximum capacity. 
Together, this gives for all $k \in P_i$ either
\[
c_{sk} = f_{sk} = \sum_{l:l\in P_i, (k,l)\in E} f_{kl}
\]
or
\[
c_{kr} = f_{kr} = \sum_{l:l\in P_i, (l,k)\in E} f_{lk}
\]
depending on if $k$ is connected to the source $s$ or the sink $r$ in $\mathcal{G}^M_{i}$
and by definition of $c_{sk}$ and $c_{kr}$ then
\[
0 = \frac{\partial h}{\partial \beta_k} + \lambda_1 w_k \text{sign}(\beta_k^{(m)}) + \sum_{l:l\in P_i, (k,l)\in E} f_{kl}
\]
for all $k \in P_i$ (noting that $f_{kl}=-f_{lk}$). 
By definition of $c_{kl}$ we know that $-\lambda_2 w_{kl} \le f_{kl} \le \lambda_2 w_{kl}$
and therefore
\[
-1 \le t_{kl} := \frac{f_{kl}}{\lambda_2 w_{kl}} \le 1
\]
for all $k,l \in P_i; (k,l) \in E$. Therefore, with this definition of $t_{kl}$, 
Equation (\ref{Eq:SubgradientWithGamma}) holds, which implies that the subgradient equations
hold.

Now it remains to show the same for $\beta_k = 0$ for $k \in P_i$ for some $i$.
However, we have already shown in Lemma \ref{Lem:Beta0} that in this case
Equation (\ref{Eq:SubgradientWithGamma}) holds, which in turn implies that the 
subgradient equations hold.

Putting all these cases together, we have shown that the subgradient equations hold
for all $k \in P_i$ and for all $1 \le i \le |\mfP|$ and therefore, the 
optimality of $\bbeta$ is shown.
\end{proof}

Now that we have established conditions for optimality of the results, we only need
two more, very brief lemmata.

\begin{lem}
Assume that we only allow fusion of sets, but not splits. Then the algorithm
converges after a finite number of steps.
\label{Lem:FiniteSteps}
\end{lem}
\begin{proof}
Assume that we are currently in step $m$ with estimate $\bbeta^{(m)}$ and associated
partition $\mfP^{(m)}=\mfF^{(m)}$ (as we only allow fusions, not splits, the equality has
to hold). Then
after the application of the coordinate-wise algorithm we get the next estimate 
$\bbeta^{(m+1)}$  with partition $\mfP^{(m+1)}$. If $\mfP^{(m+1)} = \mfP^{(m)}$,
then we cannot move any further without splits after a total of $m+1$ steps. 
However if $\mfP^{(m+1)} \neq \mfP^{(m)}$, the
fact that we can only fuse coefficients implies that $|\mfP^{(m+1)}|<|\mfP^{(m)}|$. 
Therefore, at most $|\mfP^{(m+1)}|$ further steps are possible. Thus, the algorithm
converges after a finite number of steps. 
\end{proof}

\begin{lem}
Let $\bbeta$ have partition $\mfP$ with associated transformed coefficients 
$\tilde{\bbeta}$ and transformed problem $\tilde{g}$. Then there are only
finitely many points $\bbeta$ such that $\tilde{\bbeta}$ is coordinate-wise
optimal w.r.t. $\tilde{g}$. 
\label{Lem:FiniteOptimalPoints}
\end{lem}
\begin{proof}
First of all, we note that by the definition of $\mfP$, it holds that
$\tilde{g}$ is differentiable w.r.t. $\tilde{\beta}_i$ for all $\tilde{\beta}_i \neq 0$
and as they are coordinate-wise optimal, $\partial \tilde{g} / \partial \tilde{\beta}_i = 0$. Let
\[
N = \{i : \tilde{\beta}_i = 0 \}.
\]
As the function $\tilde{g}$ is strictly convex in the variables 
$N^c = \{1, \ldots, |\mfP|\} \backslash N$ at point $\tilde{\bbeta}$ and thus,
fixing all variables in $N$ at 0, $\tilde{\bbeta}$ is the global solution for $N^c$ and
is unique. Therefore, for any fixed partition $\mfP$ and set $N$, there is at most
one $\tilde{\bbeta}$ that is coordinate-wise optimal. As there are only finitely
many $\mfP$ and $N$, the claim holds.
\end{proof}

After all these lemmata, we now show Theorem \ref{Thm:CompleteCoordinate}.
\begin{proof}
First, we note that Lemma \ref{Lem:FiniteSteps} guarantees that in a finite number of 
steps, we will be able to get to a coordinate-wise optimal point that cannot be moved
anymore by only fusing sets. Furthermore, Lemma \ref{Lem:FiniteOptimalPoints} guarantees
that there are only finitely many such points. In Lemma \ref{Lem:ConvergenceConditions},
part (\ref{Enum:CaseSplitSet}) shows that if we are at such a point and we can split
a set here, then we will move away (and never return as the coordinate-wise 
procedure always improves the loss function).
On the other hand, part (\ref{Enum:Optimum}) shows that if we cannot split a set, then
we are at the global optimum.

Overall, we see at worst, in a finite number of steps we can visit each such coordinate-wise
optimal and stable point, but will only visit each one at most once and only stop at the global
optimum, which is trivially one of these coordinate-wise optimal and stable points. 

Therefore, the algorithm converges to the global optimum in a finite number of steps.
\end{proof}

\bibliographystyle{plainnat}
\bibliography{FusedLasso}

\end{document}

%% file: definitions.tex
\def\mfF{{\mathfrak{F}}}
\def\mfP{{\mathfrak{P}}}

\def\bx{{\bf x}}

\def\bs{{\bf s}}

\def\b1{{\bf 1}}

\def\bz{{\bf z}}
\def\bw{{\bf w}}
\def\bv{{\bf v}}
\def\bX{{\bf X}}
\def\bA{{\bf A}}

\def\bD{{\bf D}}

\def\bW{{\bf W}}
\def\bI{{\bf I}}

\def\by{{\bf y}}
\def\bd{{\bf d}}
\def\bv{{\bf v}}
\def\bY{{\bf Y}}
\def\bQ{{\bf Q}}
\def\bD{{\bf D}}
\def\bbeta{{\mbox{\boldmath $\beta$}}}
\def\bdelta{{\mbox{\boldmath $\delta$}}}

\def\blambda{{\mbox{\boldmath $\lambda$}}}

%% file: FusedLasso.bbl
\begin{thebibliography}{17}
\providecommand{\natexlab}[1]{#1}
\providecommand{\url}[1]{\texttt{#1}}
\expandafter\ifx\csname urlstyle\endcsname\relax
  \providecommand{\doi}[1]{doi: #1}\else
  \providecommand{\doi}{doi: \begingroup \urlstyle{rm}\Url}\fi

\bibitem[Bertsekas(1999)]{Bertsekas1999}
{Dimitri P.} Bertsekas.
\newblock \emph{Nonlinear Programming}.
\newblock Athena Scientific, 1999.

\bibitem[Cormen et~al.(2001)Cormen, Leiserson, Rivest, and Stein]{Cormen2001}
{Thomas, H.} Cormen, {Charles E.} Leiserson, {Ronald L.} Rivest, and Clifford
  Stein.
\newblock \emph{Introduction to Algorithms}.
\newblock MIT Press and McGraw-Hill, 2nd edition, 2001.

\bibitem[Friedman et~al.(2007)Friedman, Hastie, Hoefling, and
  Tibshirani]{FHT2007}
Jerome Friedman, Trevor Hastie, Holger Hoefling, and Robert Tibshirani.
\newblock Pathwise coordinate optimization.
\newblock \emph{Annals of Applied Statistics}, 2\penalty0 (1):\penalty0
  302--332, 2007.

\bibitem[Friedman et~al.(2010)Friedman, Hastie, and Tibshirani]{FHT2010}
Jerome Friedman, Trevor Hastie, and Robert Tibshirani.
\newblock Regularized paths for generalized linear models via coordinate
  descent.
\newblock \emph{Journal of Statistical Software}, 33, 2010.

\bibitem[Grant and Boyd(2008)]{gb08}
M.~Grant and S.~Boyd.
\newblock Graph implementations for nonsmooth convex programs.
\newblock In V.~Blondel, S.~Boyd, and H.~Kimura, editors, \emph{Recent Advances
  in Learning and Control}, Lecture Notes in Control and Information Sciences,
  pages 95--110. Springer-Verlag Limited, 2008.

\bibitem[Grant and Boyd(2010)]{cvx}
M.~Grant and S.~Boyd.
\newblock {CVX}: Matlab software for disciplined convex programming, version
  1.21.
\newblock \url{http://cvxr.com/cvx}, April 2010.

\bibitem[Hastie and Tibshirani(1990)]{GeneralizedAdditiveModels}
Trevor Hastie and Robert Tibshirani.
\newblock \emph{Generalized Additive Models}.
\newblock Chapman and Hall, 1990.

\bibitem[H\"ofling()]{Hoefling:FusedLassoPackage}
Holger H\"ofling.
\newblock Fused lasso.
\newblock to be submitted to CRAN.

\bibitem[H\"ofling(2010)]{Hoefling2010}
Holger H\"ofling.
\newblock A path algorithm for the fused lasso signal approximator.
\newblock \emph{submitted}, 2010.

\bibitem[Nelder and McCullagh(1989)]{GeneralizedLinearModels}
{J.A.} Nelder and P.~McCullagh.
\newblock \emph{Generalized Linear Models}.
\newblock Cahpman and Hall, 2nd edition, 1989.

\bibitem[Rinaldo(2009)]{Rinaldo2009}
Alessandro Rinaldo.
\newblock Properties and refinements fo the fused lasso.
\newblock \emph{Annals of Statistics}, 37:\penalty0 2922--2952, 2009.

\bibitem[Tibshirani and Wang(2007)]{TW2007}
R.~Tibshirani and P.~Wang.
\newblock Spatial smoothing and hot spot detection for cgh data using the fused
  lasso.
\newblock \emph{Biostatistics}, 2007.

\bibitem[Tibshirani(1996)]{Tibshirani1996}
Robert Tibshirani.
\newblock Regression shrinkage and selection via the lasso.
\newblock \emph{Journal of the Royal Statistical Society Series B},
  58:\penalty0 267--288, 1996.

\bibitem[Tseng(2001)]{Tseng2001}
P.~Tseng.
\newblock Convergence of a block coordinate descent method for
  nondifferentiable minimization.
\newblock \emph{Journal of Optimization Theory and Applications}, 109:\penalty0
  475--494, 2001.

\bibitem[Yuan and Lin(2006)]{Yuan2006}
M.~Yuan and Y.~Lin.
\newblock Model selection and and estimation in regression with grouped
  variables.
\newblock \emph{Journal of the Royal Statistical Society Series B},
  68:\penalty0 49--67, 2006.

\bibitem[Zou(2006)]{Zou2006}
Hui Zou.
\newblock The adaptive lasso and its oracle properties.
\newblock \emph{Journal of the American Statistical Association}, 101:\penalty0
  1418--1429, 2006.

\bibitem[Zou and Hastie(2005)]{Zou2005}
Hui Zou and Trevor Hastie.
\newblock Regularization and variable selection via the elastic net.
\newblock \emph{Journal of the Royal Statistical Society Series B},
  67:\penalty0 301--320, 2005.

\end{thebibliography}
